\documentclass[final,journal,letterpaper]{IEEEtran}
\usepackage{graphicx}
\usepackage{amssymb}
\usepackage{epstopdf}
\usepackage{amsmath}
\usepackage{enumerate}
\usepackage{cite}
\usepackage{mathcomp}
\usepackage{supertabular}
\usepackage{longtable}
\usepackage{stmaryrd}
\usepackage{url}
\newcommand{\A}{{\mathcal A}}
\newcommand{\B}{{\mathcal B}}
\newcommand{\C}{{\mathcal C}}

\newcommand{\G}{{\mathcal G}}

\newcommand{\llb}{{\llbracket}}
\newcommand{\rrb}{{\rrbracket}}

\newcommand{\lbar}{{\overline{\lambda}}}
\newcommand{\wbar}{{\overline{w}}}
\newcommand{\Jbar}{{\overline{J}}}

\newcommand{\vB}{{\sf B}}

\newcommand{\vC}{{\sf C}}

\newcommand{\vg}{{\sf g}}

\newcommand{\vu}{{\sf u}}
\newcommand{\vv}{{\sf v}}

\newcommand{\supp}{{\rm supp}}

\newcommand{\DTS}{{\text{D$\Delta$S}}}
\newcommand{\GDTS}{{\text{GD$\Delta$S}}}

\newcommand{\bbZ}{{\mathbb Z}}

\newcommand{\al}{\alpha}

\newcommand{\be}{\beta}

\newcommand{\lam}{\lambda}

\newcommand{\nin}{\noindent}

\newcommand{\mz}{\mathbb{Z}_q}

\renewcommand{\le}{\leqslant}

\newcommand{\nwq}{(n,2w-1,w)_q}
\newcommand{\nwbq}{(n,2\sum{\overline{w}}-1,\overline{w})_q}
\newcommand{\swb}{\sum \overline{w}}

\newtheorem{corollary}{Corollary}[section]
\newtheorem{definition}{Definition}[section]
\newtheorem{example}{Example}[section]
\newtheorem{lemma}{Lemma}[section]

\newtheorem{proposition}{Proposition}[section]
\newtheorem{theorem}{Theorem}[section]
\newtheorem{problem}{Problem}[section]

\title{Linear Size Optimal $q$-ary Constant-Weight  Codes and Constant-Composition Codes}

\author{    Yeow~Meng~Chee,~\IEEEmembership{Senior Member,~IEEE,}~Son~Hoang~Dau,~Alan~C.~H.~Ling,~and~San~Ling
	\thanks{The research of Y. M. Chee and S. Ling is supported in part by the National Research
	Foundation of Singapore under Research Grant NRF-CRP2-2007-03.
	The research of Y. M. Chee is also supported in part by the Nanyang 
	Technological University under Research
	Grant M58110040.}
	\thanks{Y. M. Chee, S. H. Dau and S. Ling are
		with the Division of Mathematical Sciences,
		School of Physical and Mathematical Sciences,
		Nanyang Technological University,
		21 Nanyang Link,
		Singapore 637371 (email: {\tt ymchee@ntu.edu.sg},
		{\tt daus0001@ntu.edu.sg}, {\tt lingsan@ntu.edu.sg}).}
	\thanks{A. C. H. Ling is with the Department of Computer Science,
	University of Vermont, Burlington, Vermont, USA 05405 (email: {\tt aling@emba.uvm.edu}).}
}

\date{}                                           

\begin{document}
\maketitle

\begin{abstract}
\boldmath
An optimal constant-composition or constant-weight code of weight $w$
has linear size if and only if its distance $d$ is at least
$2w-1$. When $d\geq 2w$, the determination of the exact size of 
such a constant-composition or constant-weight code is trivial, but the case
of $d=2w-1$ has been solved previously only for binary and
ternary constant-composition and constant-weight codes,
and for some sporadic instances. 

This paper provides a construction for quasicyclic 
optimal constant-composition and constant-weight
codes of weight $w$ and distance $2w-1$ based on a new generalization of
difference triangle sets. As a result, the sizes of optimal constant-composition codes
and optimal constant-weight codes
of weight $w$ and distance $2w-1$ are determined for all such codes of
sufficiently large lengths.
This solves an open problem of Etzion.

The sizes of optimal constant-composition codes
of weight $w$ and distance $2w-1$ are also determined for all $w\leq 6$, except in two cases.
\end{abstract}

\begin{keywords}
\boldmath
constant-composition codes, constant-weight codes, difference triangle sets, generalized
Steiner systems, Golomb rulers, quasicyclic codes
\end{keywords}

\section{Introduction}

\PARstart{T}here are two generalizations of binary constant-weight codes as we enlarge the alphabet
beyond size two. These are the classes of constant-composition codes and 
$q$-ary constant-weight codes. 
While a vast amount of knowledge exists for binary constant-weight codes
\cite{MacWilliamsSloane:1977,Brouweretal:1990,Agrelletal:2000,Smithetal:2006}, relatively little is known about constant-composition codes and $q$-ary constant-weight codes. 
Recently, these classes of codes
have attracted  some attention
\cite{Svanstrom:2000,OstergardSvanstrom:2002,Svanstrometal:2002,BogdanovaKapralov:2003,Luoetal:2003,Chuetal:2004,Colbournetal:2004,Chuetal:2005,DingYin:2005a,DingYin:2005b,DingYuan:2005,DingYin:2006,CheeLing:2007,Cheeetal:2007,Cheeetal:2008,Cheeetal:2008b} due to
several important applications requiring nonbinary alphabets, such
as in determining the zero error decision feedback capacity of discrete memoryless channels
\cite{TelatarGallager:1990}, multiple access communications \cite{Dyachkov:1984},
spherical codes for modulation \cite{EricsonZinoviev:1995}, DNA codes
\cite{King:2003,MilenkovicKashyap:2006,CheeLing:2008}, powerline communications
\cite{Chuetal:2004,Colbournetal:2004}, frequency hopping \cite{Chuetal:2006}, and
coding for bandwidth-limited channels \cite{CostelloForney:2007}.

As in the case of binary constant-weight codes, the determination of the maximum size
of a constant-composition code or a $q$-ary constant-weight code of length $n$, given constraints on its distance, weight and/or composition,
constitutes a central problem in their investigation.


The ring $\bbZ/q\bbZ$ is denoted by $\bbZ_q$. For integers $m\leq n$,
the set of integers $\{m,m+1,\ldots,n\}$ is denoted $[m,n]$. The set $[1,n]$
is further abbreviated to $[n]$. 
A {\em partition} is a tuple $\lbar=\llb\lambda_1,\ldots,\lambda_N\rrb$
of integers such that $\lambda_1\geq\cdots\geq\lambda_N\geq 1$. 
The $\lambda_i$'s are the {\em parts} of the partition. Disjoint set union is denoted
by $\sqcup$.

If $X$ and $R$ are sets, $X$ finite, then
$R^X$ denotes the set of vectors of length $|X|$, where each
component of a vector $\vu\in R^X$ has value in $R$ and is indexed by an element of $X$,
that is, $\vu=(\vu_x)_{x\in X}$.
A {\em $q$-ary code of length $n$} is a set $\C\subseteq \bbZ_q^X$, for some
$X$ of size $n$.
The elements of $\C$ are called {\em codewords}.
The {\em support} of a vector $\vu\in \bbZ_q^X$,
denoted $\supp(\vu)$, is the set $\{x\in X: \vu_x\not=0\}$.
The {\em Hamming norm} or {\em weight}
of $\vu\in \bbZ_q^X$ is defined as $\| \vu\| = |{\rm supp}(\vu)|$. The distance
induced by this norm is called the {\em Hamming distance}, denoted $d_H(\cdot,\cdot)$, so
that $d_H(\vu,\vv)=\| \vu-\vv\|$, for $\vu,\vv\in \bbZ_q^X$.
A code $\C$ is said to have {\em distance} $d$ if
$d_H(\vu,\vv)\geq d$ for all distinct $\vu,\vv\in\C$.
The {\em composition} of a vector $\vu\in \bbZ_q^X$ is the tuple
$\overline{w}=\llb w_1,\ldots,w_{q-1}\rrb$, where $w_i=|\{x\in X:\vu_x=i\}|$, $i\in\bbZ_q\setminus\{0\}$.
A code $\C$ is said to have {\em constant weight} $w$ if every codeword in $\C$ has weight $w$,
and is said to have {\em constant composition} $\overline{w}$ if every codeword in $\C$
has composition $\overline{w}$. Hence, every constant-composition code is
a constant-weight code. We refer to a $q$-ary code of length $n$, distance $d$,
and constant weight $w$ as an $(n,d,w)_q$-code. If in addition, the code
has constant composition $\overline{w}$, then it is referred to as an $(n,d,\overline{w})_q$-code.
An $(n,d,w)_2$-code and an $(n,d,\llb w\rrb)_2$-code 
coincide in definition, and are binary constant-weight
codes. 
The maximum size of an $(n,d,w)_q$-code is denoted $A_q(n,d,w)$ and that of an
$(n,d,\overline{w})_q$-code is denoted $A_q(n,d,\overline{w})$. 
Any $(n,d,w)_q$-code or $(n,d,\overline{w})_q$-code attaining the maximum size is called {\em optimal}.

The following operations do not affect distance and composition properties
of an $(n,d,\overline{w})_q$-code:

\begin{enumerate}
\item reordering the components of $\overline{w}$, and
\item deleting zero components of $\overline{w}$.
\end{enumerate}
Consequently, throughout this paper, attention is restricted to those compositions
$\overline{w}=\llb w_1,\ldots,w_{q-1}\rrb$, where $w_1\geq \cdots\geq w_{q-1}\geq 1$,
that is, $\overline{w}$ is a partition.
For succinctness, the sum $\sum_{i=1}^{q-1} w_i$
of all the parts
of a partition $\overline{w}=\llb w_1,\ldots,w_{q-1}\rrb$ is denoted by $\sum\overline{w}$.

The focus of this paper is on determining $A_q(n,d,w)$ and $A_q(n,d,\overline{w})$ for those $d$, $w$
and $\overline{w}$ for which $A_q(n,d,w)=O(n)$ and $A_q(n,d,\overline{w})=O(n)$. 

The Johnson-type bound of Svanstr\"om for ternary constant-composition
codes \cite[Theorem 1]{Svanstrom:2000} extends easily to the following
(see also \cite[Proposition 1.3]{Chuetal:2006}):

\vskip 10pt
\begin{proposition}[Johnson Bound]
\label{johnson}
\begin{align*}
A_q(n,d,\llb w_1,w_2,&\ldots,w_{q-1}\rrb) \leq \\
& \left\lfloor \frac{n}{w_1}A_q(n-1,d,\llb w_1-1,w_2,\ldots,w_{q-1}\rrb) \right\rfloor.
\end{align*}
\end{proposition}
\vskip 10pt

The following Johnson-type bound for $q$-ary constant-weight codes was established
in \cite[Theorem 10]{OstergardSvanstrom:2002}.

\vskip 10pt
\begin{proposition}[Johnson Bound]
\label{johnson1}
\begin{align*}
A_q(n,d,w) \leq
& \left\lfloor \frac{n(q-1)}{w}A_q(n-1,d,w-1) \right\rfloor.
\end{align*}
\end{proposition}
\vskip 10pt

\begin{definition}[Refinement]
A partition $\overline{w}=\llb w_1,\ldots,w_q\rrb$ is a {\em refinement} of
$\overline{v}=\llb v_1,\ldots,v_{q'}\rrb$ (written $\overline{w}\succcurlyeq \overline{v}$)
if there exist pairwise disjoint sets $I_1,\ldots,I_{q'}\subseteq [q]$ satisfying
$\cup_{j\in[q']} I_j = [q]$ such that
$\sum_{i\in I_j} w_i=v_j$ for each $j\in[q']$.
\end{definition}
\vskip 10pt

Chu {\em et al.} \cite{Chuetal:2006} made the following observation.

\vskip 10pt
\begin{lemma}
\label{refinement}
If $\overline{w}\succcurlyeq\overline{v}$, then
$A_q(n,d,\overline{w}) \geq A_{q'}(n,d,\overline{v})$.
\end{lemma}
\vskip 10pt

Given $q$ and $w$, the condition for $A_q(n,d,\overline{w})=O(n)$ to hold can be characterized as follows.

\vskip 10pt
\begin{proposition}
$A_q(n,d,\overline{w})=O(n)$ if and only if $d\geq 2\sum\overline{w}-1$.
\end{proposition}

\begin{proof}
$A_q(n,d,\overline{w})=O(n)$ when $d\geq 2\sum\overline{w}-1$ follows easily from the Johnson bound.

R\"{o}dl's proof
\cite{Rodl:1985}
of the Erd\H{o}s-Hanani conjecture \cite{ErdosHanani:1963} implies that
$A_2(n,d,w)=(1-o(1))\binom{n}{w-d/2+1}/\binom{w}{w-d/2+1}$, so that
$A_2(n,d,w)=\Omega(n^2)$ for all $d\leq 2w-2$.
Therefore, by Lemma \ref{refinement},
$A_q(n,d,\overline{w})\geq A_2(n,d,\sum\overline{w})=\Omega(n^2)$
for all $d\leq 2\sum\overline{w}-2$.
\end{proof}
\vskip 10pt

A similar proof yields:
\vskip 10pt
\begin{proposition}
$A_q(n,d,w)=O(n)$ if and only if $d\geq 2w-1$.
\end{proposition}

\subsection{Problem Status and Contribution}

For constant-composition codes, it is trivial to see that 
\begin{equation*}
A_q(n,d,\overline{w})=
\begin{cases}
1,&\text{if $d\geq 2\sum\overline{w}+1$} \\
\lfloor n/\sum\overline{w} \rfloor,&\text{if $d=2\sum\overline{w}$.}
\end{cases}
\end{equation*}

When $d=2\sum\overline{w}-1$, our knowledge of $A_q(n,d,\overline{w})$ is limited. We know that
$A_2(n,2w-1,w)=A_2(n,2w,w)=\lfloor n/w\rfloor$, trivially.
$A_3(n,2\sum\overline{w}-1,\overline{w})$ has also been completely determined by
Svanstr\"om {\em et al.} \cite{Svanstrometal:2002}. In particular, 
$A_3(n,2\sum\overline{w}-1,\overline{w})=\lfloor n/w_1\rfloor$ holds for all $n$ sufficiently large.
Beyond this (for $q\geq 4$), $A_q(n,2\sum\overline{w}-1,\overline{w})$ has not been determined,
except in one instance: $A_4(n,5,\llb 1,1,1\rrb)=n$ for $n\geq 7$,
established by Chee {\em et al.} \cite{Cheeetal:2007}. 

For constant-weight codes, we have
\begin{equation*}
A_q(n,d,w)=
\begin{cases}
1,&\text{if $d\geq 2w+1$} \\
\lfloor n/w \rfloor,&\text{if $d=2w$.}
\end{cases}
\end{equation*}
An explicit formula for $A_3(n,2w-1,w)$ has been obtained by \"Osterg{\aa}rd and Svanstr\"om
\cite{OstergardSvanstrom:2002}.
When $q \geq 4$, the value of $A_q(n,2w-1,w)$ is not known. 

The main contribution of this paper is the following two results.

\vskip 10pt
{\em Main Theorem 1:}
\label{main1}
Let $\wbar=\llb w_1,\ldots,w_{q-1}\rrb$. Then
$A_q(n,2\sum\overline{w}-1,\overline{w})=\lfloor n/w_1\rfloor$
for all sufficiently large $n$.
\vskip 10pt

{\em Main Theorem 2:}
\label{main2}
$A_q(n,2w-1,w)= (q-1)n/w$ for all sufficiently large $n$ satisfying $w|(q-1)n$.
\vskip 10pt

\noindent In particular, Main Theorem 2 solves an open problem of 
Etzion concerning generalized Steiner systems \cite[Problem 7]{Etzion:1997}.

The optimal constant-weight and constant-composition codes constructed in the proofs of 
Main Theorem 1 and Main Theorem 2
 are quasicyclic, and are obtained from difference triangle sets and their generalization.

\section{Quasicyclic Codes}

A code is {\em quasicyclic} if there exists an $\ell$ such that
a cyclic shift of a codeword by $\ell$ places is another codeword.
More formally, let $X=\bbZ_n$ and define on $\bbZ_q^X$ the
{\em cyclic shift operator} $T:(\vu_x)_{x\in X}\mapsto(\vu_{x-1})_{x\in X}$.
A $q$-ary code $\C\subseteq \bbZ_q^X$ of length $n$ is {\em quasicyclic}
(or more precisely, {\em $\ell$-quasicyclic}) if
it is invariant under $T^\ell$ for some integer $\ell\in[n]$. If
$\ell=1$, such a code is just a cyclic code.

The following two conditions are necessary and sufficient
for a code $\C$ of constant weight $w$ to have distance $2w-1$.
\begin{description}[\setlabelwidth{(C2)}\usemathlabelsep]
\item[(C1)] For any distinct $\vu,\vv\in\C$, $|\supp(\vu)\cap\supp(\vv)|\leq 1$.
\label{C1}
\item[(C2)] For any distinct $\vu,\vv\in\C$, if $x\in \supp(\vu)\cap\supp(\vv)$, 
then $\vu_x\not=\vv_x$.
\label{C2}
\end{description}

\subsection{Quasicyclic Constant-Composition Codes}

The strategy for proving Main Theorem 1 is to construct optimal
$(n,2\sum\overline{w}-1,\overline{w})_q$-codes (meeting the Johnson bound)
that are $w_1$-quasicyclic when
$n\equiv 0\pmod{w_1}$. Optimal $(n,2\sum\overline{w}-1,\overline{w})_q$-codes,
$n\not\equiv 0\pmod{w_1}$, can be obtained easily from those with $n\equiv 0\pmod{w_1}$
by lengthening, as in the lemma below.

\vskip 10pt
\begin{lemma}[Lengthening]
\label{lengthen}
If $A_q(n,2\sum\overline{w}-1,\overline{w})=\lfloor n/w_1\rfloor$ and $n\equiv 0\pmod{w_1}$, then
$A_q(n+i,2\sum\overline{w}-1,\overline{w})=\lfloor n/w_1\rfloor$ for all $i$,
$0\leq i<w_1$.
\end{lemma}

\begin{proof}
Let $\C\subseteq\bbZ_q^X$ be an $(n,2\sum\overline{w}-1,\overline{w})_q$-code
of size $\lfloor n/w_1\rfloor$. Let $X'=X\cup\{\infty_1,\ldots,\infty_i\}$, where
$\infty_1,\ldots,\infty_i\not\in X$, and define $\C'\subseteq\bbZ_q^{X'}$ such that
$\C'=\{(c(\vu)_x)_{x\in X'}: \vu\in \C\}$, where
\begin{equation*}
c(\vu)_x = 
\begin{cases}
\vu_x,&\text{if $x\in X$} \\
0,&\text{if $x\in \{\infty_1,\ldots,\infty_i\}$.}
\end{cases}
\end{equation*}
Then $\C'$ is an $(n+i,2\sum\overline{w}-1,\overline{w})_q$-code of size $\lfloor n/w_1\rfloor$.
Since $\lfloor (n+i)/w_1\rfloor=\lfloor n/w_1\rfloor$, $\C'$ is optimal by the Johnson bound.
\end{proof}
\vskip 10pt

As opposed to lengthening a code, we can also {\em shorten} a code by selecting a position
$i$, remove those codewords with a nonzero coordinate $i$,
and deleting the $i$th coordinate from every remaining codeword.

Let $n\equiv 0\pmod{w_1}$.
A $w_1$-quasicyclic $(n,2\sum\overline{w}-1,\overline{w})_q$-code $\C$ of size $n/w_1$
can be obtained by {\em developing} a particular vector $\vg\in\bbZ_q^X$:
\begin{equation*}
\C = \{T^{w_1i}(\vg): i\in[0,n/w_1-1]\}.
\end{equation*}
Such a vector $\vg$ is called a {\em base codeword} of the quasicyclic code $\C$.
The remainder of this section develops criteria for a vector $\vg\in\bbZ_q^X$
of composition $\overline{w}$ to be a base codeword
of a $w_1$-quasicyclic $(n,2\sum\overline{w}-1,\overline{w})_q$-code $\C$,
$n\equiv 0\pmod{w_1}$.

The conditions (C1) and (C2) may be stated in terms of the base codeword $\vg$ as follows.
\begin{description}[\setlabelwidth{(C3)}\usemathlabelsep]
\item[(C3)] For $w,x,y,z\in\supp(\vg)$ such that $w\not= x$, $y\not=z$, and $\{w,x\}\not=\{y,z\}$,
we have:
\label{C3}
\begin{itemize}
\item if $x-w\equiv 0\pmod{w_1}$, then $2(x-w)\not\equiv 0\pmod{n}$;
\item if $y-w\equiv 0\pmod{w_1}$, then $x-w\not\equiv z-y\pmod{n}$.
\end{itemize}
\item[(C4)] If $\vg_x=\vg_y\not=0$, then $x-y\not\equiv 0\pmod{w_1}$.
\label{C4}
\end{description}

\subsection{Quasicyclic Constant-Weight Codes}

\begin{lemma}
\label{divlemma}
Let $n \geq w > 0$ and $q \geq 2$. Then $w|(q-1)n$ if and only if
there exist positive integers $\al$, $\be$, $\ell$ and $m$ such that
$n = \al \ell$, $w = \be \ell$, and $q - 1 = m\be$.
\end{lemma}

\begin{proof}
Assume that $w|(q-1)n$. Let $\ell = \gcd(w,n)$, and let $\al = n/\ell$, 
$\be = w/\ell$. Then $\gcd(\al,\be)=1$. Since $w|(q-1)n$, we have $\be \ell | (q-1)\al \ell$.
Hence, $\be | (q-1)$. Now let $m = (q-1)/\be$. 

The converse is obvious.    
\end{proof}
\vskip 10pt

Suppose that $w|(q-1)n$. By Lemma \ref{divlemma},
there exist positive integers $\al$, $\be$, $\ell$ and $m$ such that
$n=\al \ell$, $w=\beta \ell$, and $q-1=m\beta$. Our strategy is to construct 
$\ell$-quasicyclic optimal $\nwq$-codes of size
$(q-1)n/w=mn/\ell$ (meeting the Johnson bound). In other words, we want to find $m$
vectors, $\vg^{(1)},\ldots,\vg^{(m)}\in\mz^X$, each of weight $w$, such that
\begin{equation*}
\C=\{T^{\ell i}(\vg^{(j)}):i\in[0,n/\ell-1] \text{ and } j \in [m]\}
\end{equation*}
\nin is an $\nwq$-code of size $mn/\ell$. The vectors $\vg^{(1)},\ldots,\vg^{(m)}$ are 
referred to as {\em base codewords} of $\C$.

The conditions (C1) and (C2) can be stated in terms of the base codewords $\vg^{(1)},\ldots,\vg^{(m)}$
as follows.
\begin{description}[\setlabelwidth{(C6)}\usemathlabelsep]
\item[(C5)] Let $w,x\in \supp(\vg^{(i)})$ and $y,z\in \supp(\vg^{(j)})$ such
that $w\neq x$, $y\neq z$, and $\left\{ w,x\right\} \neq \left\{ y,z\right\} $ if $j=i$. Then we have:
\label{C5}
\begin{itemize}
\item if $x-w\equiv 0\pmod{\ell}$, then $2(x-w)\not\equiv 0\pmod{n}$;
\item if $y-w\equiv 0\pmod{\ell}$, then $x-w\not\equiv z-y\pmod{n}$.
\end{itemize}
\item[(C6)] If $\vg_{z}^{(j)}=\vg_{y}^{(j)}\neq 0$ and $z\neq y$, then 
$z-y\not\equiv 0\pmod{\ell}$, for all $j\in \lbrack m]$.
\label{C6}
\item[(C7)] If $\vg_{z}^{(i)}=\vg_{y}^{(j)}\neq 0$ ($z$ and $y$ are not
necessarily distinct), then $z-y\not\equiv 0\pmod{\ell}$, for all $i,j\in
\lbrack m]$, $i\neq j$.
\label{C7}
\end{description}

\section{A New Combinatorial Array}

Conditions (C3)--(C4) (respectively, (C5)--(C7)) suggest organizing the
elements of $\supp(\vg)$ (respectively, $\supp(\vg^{(1)}),\ldots,\supp(\vg^{(m)})$) 
of those quasicyclic constant-composition codes (respectively, constant-weight codes) into a 
two-dimensional array, with respect to their congruence class
modulo $w_1$ (respectively, $\ell$) and the value of their corresponding components in $\vg$ 
(respectively, $\vg^{(1)},\ldots,\vg^{(m)}$).

\vskip 10pt
\begin{definition}
\label{Array def}
Let $\lbar=\llb \lambda_1,\ldots,\lambda_N\rrb$ be a
partition.
A {\em $\lbar$-array} is a $\lambda_1\times N$ array $\vB$ with rows indexed by
$i\in [\lambda_1]$ and columns indexed by $j\in[N]$, such that
\begin{description}[\setlabelwidth{(P3)}\usemathlabelsep]
\item[(P1)] each cell is either empty or contains a nonnegative
integer congruent to its row index modulo $\lambda_1$;
\item[(P2)] the number of nonempty cells in column $j$ is $\lambda_j$;
\item[(P3)] if $B_i=\{b_{i,1},\ldots,b_{i,N_i}\}$ is the set of entries in row $i$ of $\vB$,
then the differences $b_{i,j}-b_{i,j'}$, $i\in [N]$, $1\leq j'\not=j\leq N_i$, are
all nonzero and distinct.
\end{description}
\end{definition}
\vskip 10pt

The {\em scope} of $\vB$ is
\begin{equation*}
\begin{split}
\sigma(\vB)=\max_{1\leq i\leq \lambda_1}  ( &\{b_{i,j}-b_{i,j'}: 1\leq j'\not=j\leq N_i\} \\ 
&\cup  \left\{\left\lceil b_{i,j}/2 \right\rceil: j\in[N_i] \right\} ).
\end{split}
\end{equation*}
In particular, if $\lambda_1=\cdots=\lambda_N=\lambda$, then a $\lbar$-array $\vB$
has all cells nonempty, and is referred to as a $(\lambda,N)$-array. 
From the definition, it is easy to see that the entries of a $\lbar$-array are all distinct. 

\vskip 10pt
\begin{example}
\label{322}
A $\llb 3,2,2\rrb$-array of scope 15:
\begin{equation*}
\begin{array}{|c|c|c|}
\hline
1 & 7 & 16 \\
\hline
2 & & 14 \\
\hline
0 & 3 & \\
\hline
\end{array}
\end{equation*}
\end{example}
\vskip 10pt

\begin{example}
\label{24}
A $(2,4)$-array of scope 42:
\begin{equation*}
\begin{array}{|c|c|c|c|}
\hline
19 & 23 & 35 & 61 \\
\hline
0 & 6 & 20 & 30 \\
\hline
\end{array}
\end{equation*}
\end{example}
\vskip 10pt

\begin{proposition}
\label{array}
Let $\overline{w}=\llb w_1,\ldots,w_{q-1}\rrb$.
If there exists a $\overline{w}$-array $\vB$, then there exists a
$w_1$-quasicyclic optimal $(n,2\sum\overline{w}-1,\overline{w})_q$-code
for all $n\equiv 0\pmod{w_1}$, $n\geq 2\sigma(\vB)+1$.
\end{proposition}

\begin{proof}
Let $\vB$ be a $\overline{w}$-array and let $C_j$ denote the set of entries in
column $j$ of $\vB$, $j\in[q-1]$. 
Define a vector $\vg\in\bbZ_q^{\bbZ_n}$, $n\geq 2\sigma(\vB)+1$, as follows:
\begin{equation*}
\vg_x = \begin{cases}
j,&\text{if $x\in C_j$} \\
0,&\text{otherwise}.
\end{cases}
\end{equation*}
Then $\vg$ has composition $\overline{w}$ and satisfies conditions (C3) and (C4).
Therefore $\vg$ is a base codeword of a $w_1$-quasicyclic optimal 
$(n,2\sum\overline{w}-1,\overline{w})_q$-code.
\end{proof}
\vskip 10pt

\begin{example}
The $\llb 3,2,2\rrb$-array in Example \ref{322} gives the base codeword
\begin{equation*}
\vg = 111200020000003030^{n-17}
\end{equation*}
for a 3-quasicyclic optimal $(n,13,\llb 3,2,2\rrb)_4$-code when $n\equiv 0\pmod{3}$, $n\geq 33$.
\end{example}
\vskip 10pt

\begin{proposition}
\label{Array and code}Suppose that $w=\beta \ell$ and $q-1=m\beta$. If there
exists an $(\ell,q-1)$-array $\vB$, 
then there exists an $\ell$-quasicyclic optimal $(n,2w-1,w)_{q}$-code of
size $(q-1)n/w=mn/\ell$, provided that $\ell |n$ and $n\geq 2\sigma(\vB)+1.$
\end{proposition}

\begin{proof}
Let $\vB$ be an $(\ell,q-1)$-array and let $C_{i}$ denote the set of entries in
column $i$ of $\vB$, $i\in[q-1]$. We define the $m$ vectors $\vg^{(1)},\ldots ,\vg^{(m)}$ as
follows: for $j\in[m]$ and $0\leq z\leq n-1$,
\begin{equation}
\vg_{z}^{(j)}=
\begin{cases}
r, & \text{if $z\in C_{r}$ for some $r\in[(j-1)\beta+1,j\beta]$} \\ 
0, & \text{otherwise.}
\label{T8}
\end{cases}
\end{equation}
Since the entries of $\vB$ are distinct, $\vg^{(j)}$ 
is well-defined.
Moreover, the set of nonzero entries of $\vg^{(j)}$ is precisely
$[(j-1)\beta+1,j\beta]$, and by property (P2), each symbol in $[(j-1)\beta+1,j\beta]$
occurs exacly $\ell$ times in $\vg^{(j)}$. Therefore, 
$\vg^{(j)}\in \bbZ_q^{\bbZ_n}$ and has weight $w=\beta\ell$.

We claim that
the $m$ vectors $\vg^{(1)},\ldots ,\vg^{(m)}$ satisfy conditions (C5)--(C7), 
and hence form the base codewords for an $\ell$-quasicyclic optimal $(n,2w-1,w)_{q}$-code.
The following establishes this claim.

First, suppose that $i\not= j$. If $\vg_{z}^{(i)}$ and $\vg_{y}^{(j)}$ are nonzero, then 
$\vg_{z}^{(i)}\in [(i-1)\beta +1,i\beta]$ and $\vg_{y}^{(j)}\in [(j-1)\beta +1,j\beta]$. 
Since $i\not= j$, we have $\vg_{z}^{(i)}\not= \vg_{y}^{(j)}$. Therefore (C7)
is satisfied.

Next, suppose that $z\not= y$ and $\vg_{z}^{(j)}=\vg_{y}^{(j)}=r\not= 0$. By (\ref{T8}), 
$z,y\in C_{r}$. Since $z\not= y$, $z$ and $y$ must belong to different rows
of $\vB$. Therefore, $z\not\equiv y\pmod{\ell}$ by (P1). Thus, $\vg^{(1)},\ldots ,\vg^{(m)}$ satisfy (C6).

Now suppose that $w,x\in \supp(\vg^{(i)})$, $w\not= x$. By (\ref{T8}), there exist $r_{w}$
and $r_{x}$ such that $w\in C_{r_{w}}$ and $x\in C_{r_{x}}$. 
If $x-w\equiv 0\pmod{\ell}$, then by (P1), $x$ and $w$ are in the same row of $\vB$. 
Therefore,
\begin{equation*}
0<|x-w| \leq \sigma(\vB),
\end{equation*}
and hence,
\begin{equation*}
0<2|x-w| \leq 2\sigma(\vB)<1+2\sigma(\vB)\leq n.
\end{equation*}
It follows that $2(x-w)\not\equiv 0\pmod{n}$.

Let $w,x\in \supp(\vg^{(i)})$ and $y,z\in \supp(\vg^{(j)})$, where $w\not= x$, $y\not= z$ such that 
$y-w\equiv 0\pmod{\ell}$, and if $i=j$ then $\left\{ w,x\right\} \not=
\left\{ y,z\right\} $. We want to show that 
\begin{equation*}
x-w\not\equiv z-y\pmod{n}\text{,}
\end{equation*}
or equivalently,
\begin{equation}
y-w\not\equiv z-x\pmod{n}\text{.}  \label{T9}
\end{equation}
Again, by (\ref{T8}), $w$, $x$, $y$, and $z$ are entries of $\vB$. 
Moreover, $w$ and $y$ are in the same row. We consider two cases. 

\begin{description}[\setlabelwidth{Case $\ y\neq w$:}\usemathlabelsep]
\item[Case $w\not=y$]
Since $0<|y-w|\leq \sigma(\vB)<n$, we have $y-w\not \equiv 0\pmod{n}$. Therefore, if $x=z$, 
then (\ref{T9}) holds. If $x\not= z$ and both $x$ and 
$z$ are in the same row, then (\ref{T9}) holds by property (P3) 
of $\vB$ and the assumption that $y\not= z$ and $n\geq 2\sigma(\vB)+1$. If $x$ and $z$
are in different rows, then by (P1), $z-x\not\equiv 0\pmod{\ell}$.
Since $y-w\equiv 0\pmod{\ell}$ and $\ell |n$, (\ref{T9}) follows. 

\item[Case $w=y$]
We claim that $i=j$. Indeed, assume that $y\in C_{r_{y}}$ and $w\in
C_{r_{w}} $. Then $r_{y}\in [ (j-1)\beta +1,j\beta] $ and $
r_{w}\in [ (i-1)\beta +1,i\beta] $. Hence, if $i\neq j$, then $
r_{y}\neq r_{w}$. Therefore, there are two entries in different columns of $\vB$ that
have the same value $y$, which is a contradiction. Hence, $i=j$. Since 
$\left\{ w,x\right\} \neq \left\{ y,z\right\} $, we have $x\not= z$.
Therefore, (\ref{T9}) holds.
\end{description}

Consequently, $\vg^{(1)},\ldots ,\vg^{(m)}$ satisfy (C5).
\end{proof}
\vskip 10pt 

\begin{example}
The $(2,4)$-array of scope $42$ in Example \ref{24}
gives $\vg^{(1)}$ and $\vg^{(2)}$, where
\begin{align*}
\vg_{z}^{(1)} &=\begin{cases}
1, & \text{if $z\in \left\{ 0,19\right\}$} \\ 
2, & \text{if $z\in \left\{ 6,23\right\}$} \\ 
0, & \text{otherwise,}
\end{cases}
\\
\vg_{z}^{(2)} &=\begin{cases}
3, & \text{if $z\in \left\{ 20,35\right\}$} \\ 
4, & \text{if $z\in \left\{ 30,61\right\}$} \\ 
0, & \text{otherwise.}
\end{cases}
\end{align*}
In this case, $q=5$, $w=4$,
$\beta=2$, $\ell=2$, and $m=2$. The vectors $\vg^{(1)}$ and $\vg^{(2)}$ form 
the base codewords of a $2$-quasicyclic optimal $(n,7,4)_{5}$-code when $n$ is even and $n\geq
85=2\times 42+1$.
\end{example}
\vskip 10pt

In view of Proposition \ref{array} and Proposition \ref{Array and code}, to prove 
Main Theorem 1 and Main Theorem 2,
it suffices to construct a $\overline{\lambda}$-array
for every partition $\overline{\lambda}$.

\section{Generalized Difference Triangle Sets}

In this section, the concept of {\em difference triangle sets} is generalized and used
to produce $\lbar$-arrays.
We begin with the definition of a difference triangle set.

\vskip 10pt
\begin{definition}
An $(I,J)$-{\em difference triangle set} ($\DTS$)
is a set $\A=\{A_1,\ldots,A_I\}$, where
$A_i=\{a_{i,1},\ldots,a_{i,J}\}$, $0=a_{i,1}<\cdots<a_{i,J}$, are lists of integers such 
that the differences
$a_{i,j}-a_{i,j'}$, $i\in[I]$, $1\leq j'\not=j\leq J$, are all distinct.
\end{definition}
\vskip 10pt

\begin{example}
\label{33}
A $(3,4)$-$\DTS$:
\begin{equation*}
\{\{0,1,10,18\}, \{0,2,7,13\}, \{0,3,15,19\}\}.
\end{equation*}
The corresponding differences are displayed in triangular arrays below:
\begin{equation*}
\begin{array}{ccc}
1 & 10 & 18 \\
& 9 & 17 \\
& & 8 \\
\end{array}
\ \ \ \ \ \ \ \
\begin{array}{ccc}
2 & 7 & 13 \\
& 5 & 11 \\
& & 6 \\
\end{array}
\ \ \ \ \ \ \ \ 
\begin{array}{ccc}
3 & 15 & 19 \\
& 12 & 16 \\
& & 4 \\
\end{array}
\end{equation*}
\end{example}
\vskip 10pt

The {\em scope} of an $(I,J)$-$\DTS$ $\A=\{A_1,\ldots,A_I\}$ is
\begin{equation*}
m(\A) = \max_{A\in\A} \{a \in A\}.
\end{equation*}
Difference triangle sets with scope as small as possible are often required for applications.
Define
\begin{equation*}
M(I,J) = \min \{ m(\A) : \text{$\A$ is an $(I,J)$-$\DTS$} \}.
\end{equation*}
Difference triangle sets were introduced by Kl{\o}ve \cite{Klove:1988,Klove:1989}
and have numerous applications
\cite{Babcock:1953,Eckler:1960,RobinsonBernstein:1967,Biraudetal:1974,Blumetal:1975,FangSandrin:1977,Atkinsonetal:1986}.
A $(1,J)$-$\DTS$ is known as a {\em Golomb ruler}
with $J$ marks. 

We generalize difference triangle sets as follows.

\vskip 10pt
\begin{definition}
Let $\Jbar=\llb J_1,\ldots,J_I\rrb$ be a partition. 
A set $\A=\{A_1,\ldots,A_I\}$ with $A_i=\{a_{i,1},\ldots,a_{i,J_i}\}$,
$0= a_{i,1}<\cdots<a_{i,J_i}$, is a
$\Jbar$-{\em generalized difference triangle set} ($\GDTS$)
if the differences $a_{i,j}-a_{i,j'}$, $i\in[I]$, $1\leq j'\not=j\leq J_i$, are all distinct.
\end{definition}
\vskip 10pt
Thus, a $\GDTS$ is similar to a $\DTS$, but allowing the sets to be of different sizes.
In particular, if $J_1=\cdots=J_I=J$, then a $\Jbar$-$\GDTS$ is an $(I,J)$-$\DTS$.
The scope of a $\GDTS$ $\A=\{A_1,\ldots,A_I\}$ is defined similarly as for a $\DTS$:
\begin{equation*}
m(\A) = \max_{A\in\A} \{a \in A\}.
\end{equation*}

We now relate $\Jbar$-$\GDTS$ to $\lbar$-arrays. Let
$\lbar=\llb\lambda_1,\ldots,\lambda_N\rrb$ be a partition.
The {\em Ferrers diagram} of $\lbar$ is an array of cells with $N$ left-justified rows
and $\lambda_i$ cells in row $i$. The {\em conjugate} of $\lbar$ is the
partition $\overline{\lambda^*}=\llb \lambda_1^*,\ldots,\lambda_{\lambda_1}^*\rrb$,
where $\lambda_j^*$ is the number of parts of $\lbar$ that are at least $j$.
$\overline{\lambda^*}$ can
also be obtained by reflecting the
Ferrers diagram of $\lbar$ along its main diagonal. Conjugation of partitions is
an involution.

\vskip 10pt
\begin{example}
The Ferrers diagrams of the partition $\llb 5,3,3,2\rrb$ and its conjugate $\llb 4,4,3,1,1\rrb$ 
are shown respectively below:
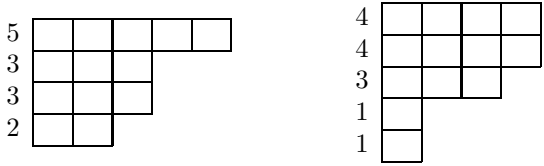
\begin{equation*}
\begin{array}{c|c|c|c|c|c|}
\cline{2-6}
5 & \phantom{*} & \phantom{*} & \phantom{*} & \phantom{*} & \phantom{*} \\
\cline{2-6}
3 & & & \\
\cline{2-4}
3 & & &  \\
\cline{2-4}
2 & &  \\
\cline{2-3} 
\end{array}
\ \ \ \ \ \ \ \ \ \ \ \ 
\begin{array}{c|c|c|c|c|}
\cline{2-5}
4 & \phantom{*} & \phantom{*} & \phantom{*} & \phantom{*} \\
\cline{2-5}
4 & & & & \\
\cline{2-5}
3 & & & \\
\cline{2-4}
1 & \\
\cline{2-2}
1 & \\
\cline{2-2}
\end{array}
\end{equation*}
\end{example}

\vskip 10pt
\begin{proposition}
\label{arrayfromGDTS}
Let $\lbar=\llb\lambda_1,\ldots,\lambda_N\rrb$ be a partition.
If there exists a $\overline{\lambda^*}$-$\GDTS$ of scope $s$, then there exists a
$\lbar$-array of scope at most $s\lambda_1$.
\end{proposition}

\begin{proof}
Let $\overline{\lambda^*}=\llb \lambda^*_1,\ldots,\lambda^*_{\lam_1}\rrb$ and let
$\A=\{A_1,\ldots,A_{\lam_1}\}$ be a $\overline{\lambda^*}$-$\GDTS$ of scope $s$.
Construct a $\lam_1\times N$ array $\vB$ as follows:
If $A_i=\{a_{i,1},\ldots,a_{i,\lambda^*_i}\}$, then the $(i,j)$th cell of $\vB$, $i\in[\lam_1]$,
$j\in[N]$,
contains $b_{i,j}=a_{i,j}\lam_1+(i~{\rm mod}~{\lambda_1})$ 
if $j\in [\lambda^*_i]$, and empty otherwise.
Then the filled cells of $\vB$ take the shape of the Ferrers diagram of $\overline{\lambda^*}$.
Thus, the number of nonempty cells in column $j$ of $\vB$ is precisely $\lambda_j$.
It is also easy to see that each entry in row $i$ of $\vB$ is congruent to 
$i~{\rm mod~}{\lambda_1}$.
The differences $b_{i,j}-b_{i,j'}$
are all distinct because
the differences $a_{i,j}-a_{i,j'}$ are all distinct in the $\GDTS$ $\A$. Moreover, all of these differences 
are at most $s\lam_1$. Finally, for any $i \in [\lam_1]$ and $j \in [\lambda^*_i]$, 
\begin{equation*}
\left\lceil \dfrac{b_{i,j}}{2}\right\rceil \le \left\lceil \dfrac{s\lam_1+(\lam_1-1)}{2}\right\rceil \le \dfrac{s\lam_1+\lam_1}{2} \le s\lam_1.
\end{equation*}
Therefore $\vB$ is a $\lbar$-array of scope at most $s\lam_1$.
\end{proof}
\vskip 10pt

\begin{corollary}
\label{arrayfromDTS}
If there exists a $(\lam,N)$-$\DTS$ of scope $s$, then there exists a $(\lam,N)$-array 
of scope at most $s\lam$.
\end{corollary}
\vskip 10pt


\begin{example}
Since $\llb 3,3,2,2\rrb^*=\llb 4,4,2\rrb$, we can construct a $\llb 3,3,2,2\rrb$-array 
from a $\llb 4,4,2\rrb$-$\GDTS$ via the proof of Proposition \ref{arrayfromGDTS}. If the 
$\llb 4,4,2\rrb$-$\GDTS$ is $\A=\{\{0,1,10,18\},\{0,2,7,13\},\{0,3\}\}$, the $\llb 3,3,2,2\rrb$-array
obtained is 
\begin{equation*}
\begin{array}{|c|c|c|c|}
\hline
1 & 4 & 31 & 55 \\
\hline
2 & 8 & 23 & 41 \\
\hline
0 & 9 & & \\
\hline
\end{array}
\end{equation*}
This array has scope 54.
\end{example}
\vskip 10pt

\begin{example}
From the $(3,4)$-$\DTS$ $\A=\{\{0,1,10,18\}$, $\{0,2,7,13\}$, $\{0,3,15,19\}\}$, we can construct 
the following $(3,4)$-array via the proof of Proposition \ref{arrayfromGDTS}.
\begin{equation*}
\begin{array}{|c|c|c|c|}
\hline
1 & 4 & 31 & 55 \\
\hline
2 & 8 & 23 & 41 \\
\hline
0 & 9 & 45 & 57 \\
\hline
\end{array}
\end{equation*}
This array has scope 57.
\end{example}
\vskip 10pt

\section{Proofs of the Main Theorems}

In this section, we use Golomb rulers to construct $\GDTS$ and provide proofs to
Main Theorem 1 and Main Theorem 2.

Let $\wp(x)$ denote the smallest prime power not smaller than $x$.
Atkinson {\em et al.} \cite[Lemma 2]{Atkinsonetal:1986} proved the following.

\vskip 10pt
\begin{theorem}
\label{atkinson}
$M(1,J) \leq (J-1)\wp(J-1)$.
\end{theorem}
\vskip 10pt

%



\begin{proposition}
\label{Existence of GDTS}
For any partition $\Jbar=\llb J_1,\ldots,J_I\rrb$, there exists a $\Jbar$-$\GDTS$ of scope
at most $(\sum\Jbar-1)\wp(\sum\Jbar-1)$.
\end{proposition}

\begin{proof}
By Theorem \ref{atkinson}, there exists a Golomb ruler $\{R\}$ of $\sum\Jbar$ marks
and scope $m(\{R\}) \leq (\sum\Jbar-1)\wp(\sum\Jbar-1)$. Partition $R$ into $I$ subsets,
$R=R_1\sqcup\cdots\sqcup R_I$, where
$|R_i|=J_i$, $i\in[I]$. Suppose 
\begin{equation*}
R_i=\{r_{i,1},\ldots,r_{i,J_i}\}, 
\end{equation*}
where $0 \leq r_{i,1}<\cdots <r_{i,J_i}$. For each $i\in[I]$, let
\begin{equation*}
A_i=\{a_{i,1},\ldots,a_{i,J_i}\},
\end{equation*}
where $a_{i,j}=r_{i,j}-r_{i,1}$, $j\in[J_i]$. Then the set 
$\A=\{A_1,\ldots,A_I\}$ forms a $\Jbar$-$\GDTS$ of scope
\begin{equation*}
m(\A)\leq m(\{R\}) \leq \left(\sum\Jbar-1\right)\wp\left(\sum\Jbar-1\right).
\end{equation*}
\end{proof}
\vskip 10pt

The following corollary is immediate. 

\vskip 10pt
\begin{corollary}
\label{Existence of DTS}
For any $I>0$ and $J>0$, there exists an $(I,J)$-$\DTS$ of scope at most
$(IJ-1)\wp(IJ-1)$.
\end{corollary}

\subsection{Proof of Main Theorem 1}
\label{MT1}


Let $\wbar=\llb w_1,\ldots,w_{q-1}\rrb$ be a partition and consider
$\overline{w^*}=\llb w^*_1,\ldots,w^*_{w_1}\rrb$.
By Proposition \ref{Existence of GDTS}, there exists a $\overline{w^*}$-$\GDTS$ of scope 
at most $(\sum \wbar -1)\wp(\sum \wbar -1)$. 
Therefore, by Proposition \ref{arrayfromGDTS}, there exists a $\wbar$-array
of scope at most $w_1(\sum \wbar -1)\wp(\sum \wbar -1)$. 
Finally, Proposition \ref{array} guarantees the existence 
of a $w_1$-quasicyclic optimal $\nwbq$-code of size $n/w_1$ for all $n \equiv 0\pmod{w_1}$,
$n\geq 2w_1(\sum \wbar -1)\wp(\sum \wbar -1)+1$. This, together with Lemma \ref{lengthen},
proves Main Theorem 1.

\subsection{Proof of Main Theorem 2}
\label{MT2}


Suppose $w|(q-1)n$. Then by Lemma \ref{divlemma}, let $w=\beta\ell$, where $\beta|(q-1)$.
By Corollary \ref{Existence of DTS}, there exists an $(\ell,q-1)$-$\DTS$ of scope at most 
$(\ell(q-1)-1)\wp(\ell(q-1)-1)$.
Therefore, by Corollary \ref{arrayfromDTS}, there exists an $(\ell,q-1)$-array of scope at most 
$\ell(\ell(q-1)-1)\wp(\ell(q-1)-1)$.
Finally, Proposition \ref{Array and code} guarantees 
the existence of an $\ell$-quasicyclic optimal $\nwq$-code of size 
$(q-1)n/w$ for all $n \equiv 0 \pmod{\ell}$, $n \geq 2\ell(\ell(q-1)-1)\wp(\ell(q-1)-1)+1$.
This proves Main Theorem 2.

In particular, by taking $\beta=1$ and $\beta=w$ respectively, we have
the following results:
\begin{enumerate}[(i)]
\item There exists a $w$-quasicyclic optimal $(n,2w-1,w)_q$-code for all $n\equiv 0\pmod{w}$, 
$n\geq 2w(w(q-1)-1)\wp(w(q-1)-1)+1$.
\item If $w|(q-1)$, then there exists a cyclic optimal $(n,2w-1,w)_q$-code
for all $n\geq 2(q-2)\wp(q-2)+1$.
\end{enumerate}


\section{Resolution of an Open Problem of Etzion}

A {\em set system} is a pair $S=(X,\B)$, where $X$ is a finite set of {\em points}, and
$\B\subseteq 2^X$. The elements of $\B$ are called {\em blocks}. The {\em order}
of $S$ is the number of points, $|X|$.
If $|B|=k$ for all
$B\in\B$, then $S$ is said to be {\em $k$-uniform}.
Let $\A\subseteq 2^X$. A {\em transverse} of $\A$ is set $T\subseteq X$ such that
$|T\cap A|\leq 1$ for all $A\in\A$.
Hanani \cite{Hanani:1963} introduced the following generalization
of $t$-designs.

\vskip 10pt
\begin{definition}
An {\em ${\rm H}(n,q,w,t)$ design} is a triple $(X,\G,\B)$, where $(X,\B)$ is a $w$-uniform set system
of order $nq$, $\G=\{G_1,\ldots,G_n\}$ is a partition of $X$ into $n$ sets, each
of cardinality $q$, such that 
\begin{enumerate}[(i)]
\item $B$ is a transverse of $\G$ for all $B\in\B$; and
\item each $t$-element transverse of $\G$ is contained in precisely one block of $\B$.
\end{enumerate}
\end{definition}
\vskip 10pt

From an ${\rm H}(n,q,w,t)$ design $(X,\G,\B)$, we can form a constant-weight code
$\C\subseteq\bbZ_{q+1}^n$ as follows.
Let $G_i=\{\gamma_{1,i},\gamma_{2,i},\ldots,\gamma_{q,i}\}$, where $0\not\in G_i$. The
code $\C$ has a codeword for each block. Assume $B=\{b_1,b_2,\ldots,b_w\}$ is a
block of $\B$ (this block is denoted by $\{\langle i_1,j_1\rangle,\langle i_2,j_2\rangle,\ldots,
\langle i_w,j_w\rangle\}$, where $b_s=\gamma_{j_s,i_s}$). We form the
codeword $\vu\in\C$ corresponding to $B$ as follows: for $i\in[n]$, 
\begin{equation*}
\vu_i = \begin{cases}
j,&\text{if $b_r=\gamma_{j,i}$ for some $r\in[w]$} \\
0,&\text{otherwise.}
\end{cases}
\end{equation*}
The distance of $\C$ is at least $w-t+1$. If $\C$ has distance $2(w-t)+1$,
Etzion \cite{Etzion:1997}
calls the ${\rm H}(n,q,w,t)$ design, from which $\C$ is constructed,
a {\em generalized Steiner system} ${\rm GS}(t,w,n,q)$.

It is not hard to verify that a ${\rm GS}(t,w,n,q)$ contains exactly
$q^t\binom{n}{t}/\binom{w}{t}$ blocks. By the Johnson bound, we have
\begin{equation*}
A_{q+1}(n,2(w-t)+1,w) \leq q^t\frac{\binom{n}{t}}{\binom{w}{t}}.
\end{equation*}
It follows from the above construction that if a ${\rm GS}(t,w,n,q)$ exists,
then 
\begin{equation*}
A_{q+1}(n,2(w-t)+1,w) = q^t\frac{\binom{n}{t}}{\binom{w}{t}}.
\end{equation*}

The next result establishes the converse when $\binom{w}{t}|q^t\binom{n}{t}$.

\vskip 10pt
\begin{proposition}
Suppose that $\binom{w}{t}|q^t\binom{n}{t}$. Then a ${\rm GS}(t,w,n,q)$ exists
if
\begin{equation*}
A_{q+1}(n,2(w-t)+1,w) = q^t\frac{\binom{n}{t}}{\binom{w}{t}}.
\end{equation*}
\end{proposition} 

\begin{proof}
Let $\C$ be an (optimal) $(n,2(w-t)+1,w)_{q+1}$-code of size $q^t\binom{n}{t}/\binom{w}{t}$.
Define
\begin{align*}
X &= \{ (i,j) : i\in [n] \text{ and } j\in[q] \} \\
\G &= \{G_i : i\in[n]\},
\end{align*}
where $G_i=\{ (i,j) : j\in[q]\}$. We associate with each codeword $\vu\in\C$ a block
$B^\vu\subseteq X$ as follows:
\begin{equation*}
B^\vu = \{ (i,j): \vu_i=j, i\in[n], j\in[q]\}.
\end{equation*}
Finally, let $\B=\{B^\vu:\vu\in\C\}$.

We claim that $(X,\G,\B)$ is a ${\rm GS}(t,w,n,q)$. Indeed, $|B|=w$ for all $B\in\B$,
and $|B\cap G_i|\leq 1$ for all $B\in\B$ and $i\in[n]$. Hence, it remains to show that any
$t$-element transverse of $\G$ is contained in exactly one block of $\B$. Suppose
$B^\vu$ and $B^\vv$ are two different blocks containing a particular
$t$-element transverse of $\G$. Then $|\supp(\vu)\cap\supp(\vv)| \geq t$,
implying $d_H(\vu,\vv)\leq 2(w-t) < 2(w-t)+1$, a contradiction.
Therefore, any $t$-element transverse of $\G$ is contained in at most one block, and
hence in exactly one block, since $|\B|=|\C|=q^t\binom{n}{t}/\binom{w}{t}$.
\end{proof}
\vskip 10pt

\begin{corollary}
\label{GSequivalence}
Suppose that $\binom{w}{t}|q^t\binom{n}{t}$.  Then there exists a ${\rm GS}(t,w,n,q)$
if and only if 
\begin{equation*}
A_{q+1}(n,2(w-t)+1,w) = q^t\frac{\binom{n}{t}}{\binom{w}{t}}.
\end{equation*}
\end{corollary}
\vskip 10pt

Etzion \cite[Problem 7]{Etzion:1997}
raised the following as an open problem for further research.

\vskip 10pt
\begin{problem}[Etzion]
\label{etzionproblem}
Given $k$ and $w$, show that there exists an $n_0$ such that for all $n\geq n_0$,
where $w|nk$, a ${\rm GS}(1,w,n,k)$ exists.
\end{problem}
\vskip 10pt

The following result, which is a direct consequence of Main Theorem 2 and Corollary
\ref{GSequivalence}, solves Problem \ref{etzionproblem}.

\vskip 10pt
\begin{theorem}
There exists a ${\rm GS}(1,w,n,k)$ for all sufficiently large $n$ satisfying $w|nk$.
\end{theorem}

\begin{proof}
By Main Theorem 2, we have
\begin{equation*}
A_{k+1}(n,2w-1,w)=kn/w,
\end{equation*}
for all sufficiently large $n$ satisfying $w|kn$. It follows immediately from
Corollary \ref{GSequivalence} that there also exists a ${\rm GS}(1,w,n,k)$
for all sufficiently large $n$ satisfying $w|kn$.
\end{proof}

\section{Explicit Bounds}

Main Theorem 1 and Main Theorem 2 are asymptotic statements: the hypothesis
that $n$ is sufficiently large must be satisfied. But how large must $n$ be?
More precisely, for a partition
$\overline{w}=\llb w_1,\ldots,w_{q-1}\rrb$ and a positive integer $w$,
define
\begin{align*}
N_{\rm ccc} &(\overline{w}) =\\
\min&\left\{n_0:A_q\left(n,2\sum{\overline{w}}-1,\overline{w}\right)=
\left\lfloor \frac{n}{w_1}\right\rfloor \text{ for all }
n\geq n_0\right\}.
\end{align*}
and
\begin{align*}
N_{\rm cwc}&(w) = \\
\min&\biggl\{n_0:A_q(n,2w-1,w)=\frac{(q-1)n}{w} \text{ for all }
n\geq n_0  \\
& \phantom{\biggl\{n_0:A_q(n,2w-1,w)=aaa}\text{ satisfying } w|(q-1)n \biggr\}.
\end{align*}
We give explicit bounds on $N_{\rm ccc}(\overline{w})$ and $N_{\rm cwc}(w)$ in this section.

\subsection{Bounds on $N_{\rm ccc}(\overline{w})$}

The proof of Main Theorem 1 in Section \ref{MT1} shows that
\begin{equation}
\label{ub}
N_{\rm ccc}(\overline{w})\leq 2w_1(\sum\overline{w}-1)\wp(\sum\overline{w}-1)+1.
\end{equation}
By Bertrand's postulate, $\wp(x)\leq 2x$ for all $x\geq 1$. For $x$ sufficiently large, better
asymptotic bounds on $\wp(x)$ exist (see for example, \cite{Bakeretal:2001}), but we are after quantifiable bounds. This implies
\begin{equation*}
N_{\rm ccc}(\overline{w}) \leq 4w_1\left(\sum\overline{w}-1\right)^2+1.
\end{equation*}
We now prove a lower bound on $N_{\rm ccc}(\overline{w})$.

\vskip 10pt
\begin{proposition}
\label{lb}
Let $\overline{w}=\llb w_1,\ldots,w_{q-1}\rrb$ be a partition.
If $w_1|n$ and there exists an $\nwbq$-code of size $n/w_1$, then 
$n \geq w_1^2k(k-1)+w_1$, where $k=\left \lfloor \sum \wbar/w_1\right \rfloor$.
In particular, when $w_1=w_2=\cdots=w_{q-1}$, we have $n \geq w_1+w_1^2(q-1)(q-2)$.
\end{proposition}

\begin{proof}
Let $\C=\{\vu^{(1)},\ldots,\vu^{(n/w_1)}\}$ be an $\nwbq$-code of size $n/w_1$. 
Then $\C$ can be regarded as
an $n/w_1\times n$ matrix $\vC$, 
whose $i$th row is $\vu^{(i)}$, $i\in[n/w_1]$. Let $N_i$
be the number of nonzero entries in column $i$ of $\vC$. 
Then $\sum_{i=1}^n N_i=(n\swb)/w_1$. 
In each column of $\vC$, we associate each pair of distinct nonzero entries with the pair of rows
that contain these entries. There are $\binom{N_i}{2}$ such pairs of nonzero entries in column 
$i$ of $\vC$. Therefore, there are $\sum_{i=1}^n \binom{N_i}{2}$ such pairs in all the
columns of $\vC$. Since there are no pairs of distinct
codewords in $\C$ whose supports intersect in two 
elements, the $\sum_{i=1}^n \binom{N_i}{2}$ pairs of rows associated with the 
$\sum_{i=1}^n \binom{N_i}{2}$ pairs of distinct nonzero entries are also all distinct. Hence,
\begin{equation*}
\sum_{i=1}^n \binom{N_i}{2} \leq \binom{|\C|}{2}=\binom{n/w_1}{2},
\end{equation*}
or equivalently,
\begin{equation}
\sum_{i=1}^n N_i(N_i-1) \leq \dfrac{n(n-w_1)}{w_1^2}.
\label{T20}
\end{equation}
Since $k=\lfloor \sum \wbar/w_1 \rfloor=\lfloor ((n\sum \wbar)/w_1)/n \rfloor$, there exists
$r\in[0,n-1]$ such that
\begin{equation*}
\dfrac{n\swb}{w_1}=kn + r. 
\end{equation*}
As $\sum_{i=1}^n N_i=(n\swb)/w_1$ we have
\begin{align}
\sum_{i=1}^n N_i(N_i-1) &\geq r(k+1)k+ (n-r)k(k-1) \nonumber \\
&\geq nk(k-1).
\label{T21}
\end{align}
From (\ref{T20}) and (\ref{T21}), we have
\begin{equation*}
\dfrac{n(n-w_1)}{w_1^2} \geq nk(k-1),
\end{equation*}
giving $n \geq w_1^2k(k-1)+w_1$.
\end{proof}

\vskip 10pt
\begin{corollary}
\label{lub}
\begin{align*}
\left(\sum\overline{w}\right)^2-w_1\left(\sum\overline{w}-1\right) &
\leq \\
N&_{\rm ccc}(\overline{w}) \\
&\leq 4w_1\left(\sum\overline{w}-1\right)^2+1.
\end{align*}
\end{corollary}
\vskip 10pt

The upper and lower bounds on $N_{\rm ccc}(\overline{w})$ in
Corollary \ref{lub} differ approximately by a
factor of $4w_1$.

\subsection{Bounds on $N_{\rm cwc}(w)$}


The proof of Main Theorem 2 in Section \ref{MT2} shows that
$N_{\rm cwc}(w)\leq 2w(w(q-1)-1)^2+1$.

For constant-weight codes, the following result of 
Etzion \cite[Theorem 1]{Etzion:1997} gives $N_{\rm cwc}(w)\geq (w-1)(q-1)+1$.

\vskip 10pt
\begin{proposition}
\label{Etzionbound}
Given $q$ and $w$, if there exists an optimal $\nwq$-code of size $(q-1)n/w$, 
then $n \geq (w-1)(q-1)+1$. 
\end{proposition}
\vskip 10pt

There is a considerable gap between these upper and lower bounds on 
$N_{\rm cwc}(w)$.
However, when $w|n$, a better upper bound can be obtained. We describe the construction
below. 

The idea of the construction is similar to the idea of the previous ones. We
determine $q-1$ base codewords, denoted $\vg^{(1)},\ldots ,\vg^{(q-1)}$, for which the
$(n/w)$-quasicyclic code
\begin{equation*}
\mathcal{C}=\{ T^{wj}(\vg^{(i)}):i\in[q-1],j\in[0,n/w-1]\} .
\end{equation*}
is an $(n,2w-1,w)_q$-code.
Let us write $\vu \stackrel{\text{T}}{\leftarrow} \vg^{(i)}$ if $\vu=T^{wj}(\vg^{(i)})$ for some $j$.
Suppose that $\vg^{(i)}\in\{0,i\}^n$, $i\in[q-1]$. 
Then $\mathcal{C}$ is an $(n,2w-1,w)_{q}$-code if the following two conditions hold.

\begin{description}[\setlabelwidth{(C8)}\usemathlabelsep]
\item[(C8)] $|\supp(\vu,\vv)|=0$ if 
$\vu \stackrel{\text{T}}{\leftarrow}\vg^{(i)}$ and 
$\vv \stackrel{\text{T}}{\leftarrow} \vg^{(i)}$ for some $i$.
\item[(C9)] $|\supp(\vu,\vv)| \leq 1$ if $\vu \stackrel{\text{T}}{\leftarrow}\vg^{(i)}$ and
$\vv \stackrel{\text{T}}{\leftarrow} \vg^{(j)}$ for $i\neq j$.
\end{description}

We observe that (C8) holds immediately if for every $i\in[q-1]$,
$\vg^{(i)}$ is chosen so that $\supp(\vg^{(i)})$ contains $w$ elements which are
congruent to $0,1,\ldots ,w-1 \pmod{w}$, respectively.

\vskip 10pt
\begin{theorem}
\label{MyTheorem3}
If $w|n$ and $n\geq w((w-1)(q-2)+1)$,
then 
$A_{q}(n,2w-1,w)=(q-1)n/w$.
\end{theorem}

\begin{proof}
It suffices to show that there exists an $(n,2w-1,w)_{q}$-code of size 
$(q-1)n/w$ for any $n\geq w((w-1)(q-2)+1)$, $n \equiv 0\pmod{w}$. We
construct $q-1$ base codewords $\vg^{(1)},\ldots ,\vg^{(q-1)}$ for such a code
as follows. For $i\in[q-1]$, $\vg^{(i)}\in\{0,i\}^n$ satisfies
\begin{equation}
\begin{split}
\supp(\vg^{(i)})=\{ 0, 1+(i-&1)w, 2+2(i-1)w,\ldots ,\\
&(w-1)+(w-1)(i-1)w \}.
\end{split}
\label{T10}
\end{equation}

Condition (C8) is satisfied immediately.
It remains to show that these $q-1$ base codewords satisfy
(C9). We prove this by contradiction. Assume that there exist 
$\vu=T^{kw}(\vg^{(i)})$ and $\vv=T^{lw}(\vg^{(j)})$, $i\neq j$, so that 
$|\supp(\vu,\vv)|\geq 2$. Suppose that $a,b\in \supp(\vu,\vv)$ and 
$a\equiv x\pmod{w}$, $b\equiv y\pmod{w}$. By (\ref{T10}) we have
\begin{align*}
a &=x+x(i-1)w+kw\pmod{n} \\
&=x+x(j-1)w+\ell w\pmod{n},
\end{align*}
and 
\begin{align*}
b &=y+y(i-1)w+kw\pmod{n} \\
&=y+y(j-1)w+\ell w\pmod{n},
\end{align*}
where the terms $kw$ and $\ell w$ result from the cyclic shift operations
applied on $\vg^{(i)}$ and $\vg^{(j)}$. These equations imply
\begin{equation*}
xw(i-j)+(k-\ell)w\equiv 0\pmod{n}
\end{equation*}
and 
\begin{equation*}
yw(i-j)+(k-\ell)w\equiv 0\pmod{n},
\end{equation*}
which together yield
\begin{equation}
(x-y)(i-j)\equiv 0\pmod{n/w}.  \label{T11}
\end{equation}
However, since $0\leq x\neq y\leq w-1$ and $1\leq i\neq j\leq q-1$, we have
\begin{equation}
0<\left\vert (x-y)(i-j)\right\vert \leq (w-1)(q-2)<n/w,  \label{T12}
\end{equation}
as $n\geq w(1+(w-1)(q-2))$. Thus, (\ref{T11}) and (\ref{T12}) lead to a
contradiction.
\end{proof}
\vskip 10pt

\section{Tables for Small-Weight Constant-Composition Codes}

\begin{table*}
\caption{Linear size optimal $(n,2\sum\overline{w}-1,\overline{w})_q$-codes of weight at most six}
\label{codetable}
\centering
\begin{tabular}{| c | c | l | l | l | c | l |}
\hline
Weight & Distance & Composition $\overline{w}$ & Base codeword & Condition on length $n$ & Size & Remark \\
\hline
2 & 3 & $\llb 1,1\rrb$ & $12$ & $n\geq 3$ & $n$ & Trivial \\
\hline
3 & 5 & $\llb 2,1\rrb$ & $112$ & $n\geq 5$ & $\lfloor n/2 \rfloor$ & Trivial \\
& & $\llb 1,1,1\rrb$ & $1203$ & $n\geq  7$ & $n$ & \cite{Cheeetal:2007} \\
\hline
4 & 7 & $\llb 3,1\rrb$ & $1112$ & $n\geq 7$ & $\lfloor n/3 \rfloor$ & Trivial \\
& & $\llb 2,2\rrb$ & $112002$ & $n\geq 10$ &  $\lfloor n/2 \rfloor$ & This paper \\
& & $\llb 2,1,1\rrb$ & $112003$ & $n\geq 10$ & $\lfloor n/2 \rfloor$ & Refinement of $\llb 2,2\rrb$ \\
& & $\llb 1,1,1,1\rrb$ & $1200304$ & $n\geq 13$ & $n$ & This paper \\
\hline
5 & 9 & $\llb 4,1\rrb$ & $11112$ & $n \geq 9$ & $\lfloor n/4 \rfloor$ & Trivial\\
& & $\llb 3,2\rrb$ & $110200020001$ & $n \geq 15$ & $\lfloor n/3 \rfloor$ &This paper \\
& & $\llb 3,1,1\rrb$ & $110200030001$ & $n \geq 15$ & $\lfloor n/3 \rfloor$ & Refinement of $\llb 3,2\rrb$ \\
& & $\llb 2,2,1\rrb$ & $100120000203$ & $n \geq 18$ & $\lfloor n/2 \rfloor$ & This paper \\
& & $\llb 2,1,1,1\rrb$ & $100120000304$ & $n \geq 18$ & $\lfloor n/2 \rfloor$ & Refinement of $\llb 2,2,1\rrb$ \\
& & $\llb 1,1,1,1,1\rrb$ & $120030000405$ & $n\geq 23$ & $n$ & This paper \\
& & & $12003000000000405$ & $n = 21$ & 21 & This paper\\
\hline
6 & 11 & $\llb 5,1\rrb$ & $111112$ & $n \geq 11$ & $\lfloor n/5 \rfloor$ & Trivial \\
& & $\llb 4,2\rrb$ & $1111200002$ & $n \geq 20$ & $\lfloor n/4 \rfloor$ & This paper \\
& & $\llb 4,1,1\rrb$ & $1111200003$ & $n \geq 20$ & $\lfloor n/4 \rfloor$ & Refinement of $\llb 4,2\rrb$ \\
& & $\llb 3,3\rrb$  & $111200020002$ & $n \geq 21$ & $\lfloor n/3 \rfloor$ & This paper\\
& & $\llb 3,2,1\rrb$ & $111200020003$ & $n \geq 21$ & $\lfloor n/3 \rfloor$ & Refinement of $\llb 3,3\rrb$ \\
& & $\llb 3,1,1,1\rrb$ & $111200030004$ & $n \geq 21$ & $\lfloor n/3 \rfloor$ & This paper\\
& & $\llb 2,2,2\rrb$ & $1120020030000003$ & $n \geq 30$ or $n = 26$ & $\lfloor n/2 \rfloor$ & This paper\\
& & $\llb 2,2,1,1\rrb$ & $1120020030000004$ & $n \geq 30$ or $n = 26$ & $\lfloor n/2 \rfloor$ & Refinement of $\llb 2,2,2\rrb$ \\
& & $\llb 2,1,1,1,1\rrb$ & $1120030040000005$ & $n \geq 30$ or $n = 26$ & $\lfloor n/2 \rfloor$ & Refinement of $\llb 2,2,2\rrb$ \\
& & $\llb 1,1,1,1,1,1\rrb$ & $120030000040500006$ & $n\geq 35$ or $n = 31$ & $n$ & This paper \\
\hline
\end{tabular}
\end{table*}

\begin{table*}
\setlength{\tabcolsep}{4pt}
\caption{Sizes of some small optimal constant-composition codes with $d=2\sum\overline{w}-1$}
\label{smallcodetable}
\centering
\begin{tabular}{| l l | r | r | r | r | r | r | r | r | r | r | r | r | r |
 r | r | r | r | r | r | r | r | r | r | r | r | r | r |}
\hline
& $n$ & 6& 7& 8& 9& 10& 11& 12& 13& 14& 15& 16& 17& 18& 19& 20& 21& 22& 23& 24& 25& 26& 27& 28& 29& 30& 31& 32\\
$\wbar$ & & & & & &  & & & & & & & & & & & & & & & & & & & & & & \\
\hline
$\llb 1,1\rrb$ & & & & & & & & & & & & & & & & & & & & & & & & & & & & \\
\hline
$\llb 2,1\rrb$ & & & & & & & & & & & & & & & & & & & & & & & & & & & & \\
$\llb 1,1,1\rrb$ & & 4 &  & & & & & & & & & & & & & & & & & & & & & & & & & \\
\hline
$\llb 3,1\rrb$ & & 1 & & & & & & & & & & & & & & & & & & & & & & & & & & \\
$\llb 2,2\rrb$ & & 1& 2 & 2 & 3 & & & & & & & & & & & & & & & & & & & & & & & \\
$\llb 2,1,1\rrb$ & & 1 & 2 & 2 &  3& & & & & & & & & & & & & & & & & & & & & & &\\
$\llb 1,1,1,1\rrb$ & & 1 & 2 & 2 & 3 & 5 & 6 & 9 & & & & & & & & & & & & & & & & & & & &\\
\hline
$\llb 4,1\rrb$ & & 1 & 1 & 1 & & & & & & & & & & & & & & & & & & & & & & & & \\
$\llb 3,2\rrb$ & & 1 & 1 & 1 & 2 & 2 & 2 & 3 & 3 & 4 & & & & & & & & & & & & & & & & & & \\
$\llb 3,1,1\rrb$ & & 1 & 1 & 1 & 2 & 2 & 2 & 3 & 3 & 4 & & & & & & & & & & & & & & & & & & \\
$\llb 2,2,1\rrb$ & & 1 & 1 & 1 & 2 & 2 & 2 & 3 & 3 & 4 & 6 & 6 & 7 & & & & & & & & & & & & & & & \\
$\llb 2,1,1,1\rrb$ & & 1 & 1 & 1 & 2 & 2 & 2 & 3 & 3 & 4 & 6 & 6 & 7 & & & & & & & & & & & & & & & \\
$\llb1,1,1,1,1\rrb$ & & 1 & 1 & 1 & 2 & 2 & 2 & 3 & 3 & 4 & 6 & 6 & 7 & 9 & 12 & 16 & & 21 & & & & & & & & & & \\
\hline
$\llb 5,1\rrb$ & & 1 & 1 & 1 & 1 & 1 & & & & & & & & & & & & & & & & & & & & & & \\  
$\llb 4,2\rrb$ & & 1 & 1 & 1 & 1 & 1 & 2 & 2 & 2 & 2 & 3 & 3 & 3 & 4 & 4 & & & & & & & & & & & & & \\   
$\llb 4,1,1\rrb$ & & 1 & 1 & 1 & 1 & 1 & 2 & 2 & 2 & 2 & 3 & 3 & 3 & 4 & 4 & & & & & & & & & & & & & \\ 
$\llb 3,3\rrb$ & & 1 & 1 & 1 & 1 & 1 & 2 & 2 & 2 & 2 & 3 & 3 & 3 & 4 & 4 & 5 & & & & & & & & & & & & \\ 
$\llb 3,2,1\rrb$ & & 1 & 1 & 1 & 1 & 1 & 2 & 2 & 2 & 2 & 3 & 3 & 3 & 4 & 4 & 5 & & & & & & & & & & & & \\ 
$\llb 3,1,1,1\rrb$ & & 1 & 1 & 1 & 1 & 1 & 2 & 2 & 2 & 2 & 3 & 3 & 3 & 4 & 4 & 5 & & & & & & & & & & & & \\
$\llb 2,2,2\rrb$  & & 1 & 1 & 1 & 1 & 1 & 2 & 2 & 2 & 2 & 3 & 3 & 3 & 4 & 4 & 5 & 7 & 7 & 8 & 9 & 10 & & 13 & 14 & 14 & & & \\
$\llb 2,2,1,1\rrb$ & & 1 & 1 & 1 & 1 & 1 & 2 & 2 & 2 & 2 & 3 & 3 & 3 & 4 & 4 & 5 & 7 & 7 & 8 & 9 & 10 & & 13 & 14 & 14 & & & \\
$\llb 2,1,1,1,1\rrb$ & & 1 & 1 & 1 & 1 & 1 & 2 & 2 & 2 & 2 & 3 & 3 & 3 & 4 & 4 & 5 & 7 & 7 & 8 & 9 & 10 & & 13 & 14 & 14 & & & \\
$\llb 1,1,1,1,1,1\rrb$ & & 1 & 1 & 1 & 1 & 1 & 2 & 2 & 2 & 2 & 3 & 3 & 3 & 4 & 4 & 5 & 7 & 7 & 8 & 9 & 10 & 13 & 14 & 16 & 20 & 25 & & 31 \\
\hline
\end{tabular}
\end{table*}


In this section, we provide two tables of exact values of $A_q(n,2\sum \wbar-1,\wbar)$ 
with $\sum \wbar \leq 6$, for almost all $n$. The only undetermined values in this range
are $A_7(n,11,\llb 1,1,1,1,1,1\rrb)$ when $n\in\{33,34\}$. The following (trivial) upper bound 
happens to be very useful when we build up the tables, as it is often tight for codes of small 
lengths. 

\vskip 10pt
\begin{lemma}
\label{trivialbound}
$A_q(n,2\sum \wbar-1,\wbar) \leq A_2(n,2\sum \wbar -2, \sum \wbar)$. 
\end{lemma}
\vskip 10pt

Table \ref{codetable} provides the base codewords for quasicyclic optimal codes of  
sufficiently large lengths.
For succinctness, we do not indicate trailing zeros
at the end of each base codeword. Therefore, the base codeword $1203$, say, 
should be interpreted as $12030^{n-4}$. In order to construct these base codewords,
we use either optimal Golomb rulers or a simple computer search to 
establish the best $\lbar$-array corresponding to the codes. 
Table \ref{smallcodetable} includes the sizes of optimal codes with small 
length $n$. These two tables together give an almost complete solution for the sizes of optimal
constant-composition codes of weight at most six.

In Table \ref{smallcodetable}, 
if a cell is empty, then it means that the corresponding size is already determined in 
Table \ref{codetable}. 
The upper bound for the sizes of codes comes from either the Johnson bound or 
Lemma \ref{trivialbound}, whichever is smaller. The lower bounds come from optimal codes constructed by
hand or by a hill-climbing algorithm. 
We refer the interested reader to the Appendix
for a complete description of these optimal codes.
We note that the values of $A_3(n,2(w_1+w_2) -1,\llb w_1,w_2\rrb)$ are included for
completeness although it has been determined earlier
by \"Osterg{\aa}rd and Svanstr\"om \cite[Theorem 8]{OstergardSvanstrom:2002}.



\begin{table}
\setlength{\tabcolsep}{4pt}
\caption{$N{\rm ccc(\overline{w})}$ and bounds on $N_{\rm ccc}(\overline{w})$}
\label{boundstable}
\centering
\begin{tabular}{| c | c | l | c | c |}
\hline
& & & & Bounds on $N_{\rm ccc}(\overline{w})$ \\
Weight & Distance & Composition $\overline{w}$ & $N_{\rm ccc}(\overline{w})$ & from (\ref{ub}) and \\
& & & &  Proposition \ref{lb} \\
\hline
2 & 3 & $\llb 1,1\rrb$ & 3 & $[3,3]$ \\
\hline
3 & 5 & $\llb 2,1\rrb$ & 5 & $[5,17]$  \\
& & $\llb 1,1,1\rrb$ & 7 & $[7,9]$   \\
\hline
4 & 7 & $\llb 3,1\rrb$ & 7 & $[7,55]$   \\
& & $\llb 2,2\rrb$ & 10 & $[10,37]$  \\
& & $\llb 2,1,1\rrb$ & 10 & $[10,37]$   \\
& & $\llb 1,1,1,1\rrb$ & 13 & $[13,19]$   \\
\hline
5 & 9 & $\llb 4,1\rrb$ & 9 & $[9,129]$  \\
& & $\llb 3,2\rrb$ & 14 & $[13,97]$  \\
& & $\llb 3,1,1\rrb$ & 14 & $[13,97]$   \\
& & $\llb 2,2,1\rrb$ & 18 & $[17,65]$   \\
& & $\llb 2,1,1,1\rrb$ & 18 & $[17,65]$  \\
& & $\llb 1,1,1,1,1\rrb$ & 23 & $[21,33]$   \\
\hline
6 & 11 & $\llb 5,1\rrb$ & 11 & $[11,251]$  \\
& & $\llb 4,2\rrb$ & 18 & $[16,201]$  \\
& & $\llb 4,1,1\rrb$ & 18 & $[16,201]$  \\
& & $\llb 3,3\rrb$  & 21 & $[21,151]$  \\
& & $\llb 3,2,1\rrb$ & 21 & $[21,151]$   \\
& & $\llb 3,1,1,1\rrb$ & 21 & $[21,151]$  \\
& & $\llb 2,2,2\rrb$ & 30 & $[26,101]$  \\
& & $\llb 2,2,1,1\rrb$ & 30 & $[26,101]$   \\
& & $\llb 2,1,1,1,1\rrb$ & 30 & $[26,101]$   \\
& & $\llb 1,1,1,1,1,1\rrb$ & $\in[33,35]$ & $[31,51]$  \\
\hline
\end{tabular}
\end{table}

Table \ref{boundstable} gives the exact value of $N_{\rm ccc}(\overline{w})$ for all
$\overline{w}$ such that $\sum\overline{w}\leq 6$, except when
$\overline{w}=\llb 1,1,1,1,1,1\rrb$. We compare these values with bounds on
$N_{\rm ccc}(\overline{w})$ given by (\ref{ub}) and Proposition \ref{lb}. There is a large gap
between these bounds. It would be interesting to close this gap.
 
\section{Conclusion}

The exact sizes of optimal constant-composition and constant-weight codes having linear size
are determined for all such codes of sufficiently large lengths. In the course of establishing
these results, we introduced several new concepts, including that of
generalized difference triangle sets and showed how they
can be constructed from Golomb rulers. The results obtained in this paper
solve an
open problem of Etzion.

\appendix

Only codes of size at least five are listed here. Those optimal codes of size four or less
can be constructed easily by hand.

\setlength{\tabcolsep}{3pt}

\subsection{Weight Four Codes}

\begin{enumerate}[1)]
\item {An Optimal $(10,7,\llb 1,1,1,1\rrb)_5$-code:}

\begin{tabular}{l l l l}
0004021300 &
2103000040 &
0040000132 &
1000204003 \\
0320140000
\end{tabular}

\item {An Optimal $(11,7,\llb 1,1,1,1\rrb)_5$-code:} 

\begin{tabular}{l l l l}
30000200041 &
00100034020 &
20014003000 &
00003040102 \\
01320000004 &
04000301200
\end{tabular}

\item {An Optimal $(12,7,\llb 1,1,1,1\rrb)_5$-code:} 

\begin{tabular}{l l l}
010020043000 &
000200301004 &
120000000403 \\
200040100030 &
400301020000 &
002000430100 \\
003014000002 &
034100000020 &
000002004310
\end{tabular}
\end{enumerate}

\subsection{Weight Five Codes}

\begin{enumerate}[1)]

\item {An Optimal $(15,9,\llb 2,2,1\rrb)_4$-code:} 

\begin{tabular}{l l l}
002100200000103 &
201010003200000 &
000300000122010 \\
000021030010002 &
010002002001300 &
120000120000030
\end{tabular}

\item {An Optimal $(16,9,\llb 2,2,1\rrb)_4$-code:} 

Lengthening of an optimal $(15,9,\llb 2,2,1\rrb)_4$-code.

\item {An Optimal $(17,9,\llb 2,2,1\rrb)_4$-code:} 

\begin{tabular}{l l}
00301002000020010 &
00003210010000200 \\
10000031000200002 &
00020100002100030 \\
20000000123010000 &
00010003200002100 \\
01200020001003000 &
\end{tabular}

\item {An Optimal $(n,9,\llb 2,1,1,1\rrb)_5$-code, $n\in[15,17]$:} 

Refinement of an optimal $(n,9,\llb 2,2,1\rrb)_4$-code, $n\in[15,17]$.

\item {An Optimal $(n,9,\llb 1,1,1,1,1\rrb)_6$-code, $n\in[15,18]$:} 

Refinement of an optimal $(n,9,\llb 2,1,1,1\rrb)_4$-code, $n\in[15,18]$.

\item {An Optimal $(19,9,\llb 1,1,1,1,1\rrb)_6$-code:} 

\begin{tabular}{l l}
0045203000000000010 &
5010020040000000003 \\
0000100050034002000 &
3004000100000205000 \\
0000400000000320501 & 
0100340200500000000 \\
0503000014000000200 &
0000002301040000005 \\
4000001000205000300 &
0000010002003500040 \\
0020000005100034000 &
2300000000010040050
\end{tabular}

\item {An Optimal $(20,9,\llb 1,1,1,1,1\rrb)_6$-code:} 

\begin{tabular}{l l}
00020000500300004010 &
51000003400002000000 \\
00000005040000000132 &
00000350000001002400 \\
02100040003000000050 &
00001034000200050000 \\
00400200100000030005 &
00000010250040300000 \\
04050000000030010200 &
10000000025000043000 \\
20003100000050000040 &
03005000000000201004 \\
30200000000100400500 &
00000000001524000003 \\
00030502004000100000 &
00342000010005000000
\end{tabular}

\item {An Optimal $(22,9,\llb 1,1,1,1,1\rrb)_6$-code:} 

Lengthening of an optimal $(21,9,\llb 1,1,1,1,1\rrb)_6$-code.
\end{enumerate}

\subsection{Weight Six Codes}

\begin{enumerate}[1)]

\item {An Optimal $(20,11,\llb 3,3\rrb)_3$-code:} 

\begin{tabular}{l l}
10000000020201002010 &
00101002001020000020 \\
00022120000100000001 &
00010000202000201100 \\
01000001000002010202 &
\end{tabular}

\item {An Optimal $(20,11,\llb 3,2,1\rrb)_4$-code:}

Refinement of an optimal $(20,11,\llb 3,3\rrb)_3$-code.

\item {An Optimal $(20,11,\llb 3,1,1,1\rrb)_5$-code:} 

Refinement of an optimal $(20,11,\llb 3,3\rrb)_3$-code.

\item {An Optimal $(20,11,\llb 2,2,2\rrb)_4$-code:} 

Refinement of an optimal $(20,11,\llb 4,2\rrb)_3$-code.

\item {An Optimal $(21,11,\llb 2,2,2\rrb)_4$-code:} 

\begin{tabular}{l l}
010000332000100020000 &
033000000021020000010 \\
302010200300000000001 &
000103000210000032000 \\
200001020003000300100 &
000200001000001000323 \\
000020000000213103000 &
\end{tabular}

\item {An Optimal $(22,11,\llb 2,2,2\rrb)_4$-code:} 

Lengthening of an optimal $(21,11,\llb 2,2,2\rrb)_4$-code.

\item {An Optimal $(23,11,\llb 2,2,2\rrb)_4$-code:} 

\begin{tabular}{l l}
10020200000001000033000 &
20000031200100030000000 \\
00000003020000200110030 &
00031020000000000000123 \\
00000000013000013002002 &
01000000300023000200001 \\
00100000001330000020200 &
02302300000000101000000
\end{tabular}

\item {An Optimal $(24,11,\llb 2,2,2\rrb)_4$-code:}

\begin{tabular}{l}
300000100000200300000012 \\
030200201300000000001000 \\
010020030002000103000000 \\
000010300020030000010200 \\
003000012000100020000300 \\
200000000213003010000000 \\
000100000000020031200003 \\
100001020000000000332000 \\
001332000000001000000020 
\end{tabular}

\item {An Optimal $(25,11,\llb 2,2,2\rrb)_4$-code:} 

\begin{tabular}{l}
0000000001223100030000000 \\
0000000100002030003002100 \\
0003001000001003000020002 \\
0030000210000200100000003 \\
1000030020000002010000300 \\
0000100000100000200330200 \\
0101000030300000002000020 \\
3012300002000010000000000 \\
0020003000000000020101030 \\
0300212300010000000000000
\end{tabular}

\item {An Optimal $(27,11,\llb 2,2,2\rrb)_4$-code:}

Lengthening of an optimal $(26,11,\llb 2,2,2\rrb)_4$-code.

\item {An Optimal $(28,11,\llb 2,2,2\rrb)_4$-code:}

\begin{tabular}{l}
1100000000220000300000003000 \\
0000001102000000000030032000 \\
0000110000003003200020000000 \\
0200000001001030003200000000 \\
2010200000000000000300010003 \\
0020020003100100000000300000 \\
0000302000030020000001100000 \\
0000000200302000010000000031 \\
0031003000000002001000000020 \\
3000000010000000022013000000 \\
0002030000010000030100000200 \\
0000000000000301000002001302 \\
0300000030000000000000220110 \\
0003000320000210100000000000
\end{tabular}

\item {An Optimal $(29,11,\llb 2,2,2\rrb)_4$-code:}

Lengthening of an optimal $(28,11,\llb 2,2,2\rrb)_4$-code.

\item {An Optimal $(n,11,\llb 2,2,1,1\rrb)_5$-code, $n\in[20,29]$:}

Refinement of an optimal $(n,11,\llb 2,2,2\rrb)_4$-code, $n\in[20,29]$.

\item {An Optimal $(n,11,\llb 2,1,1,1,1\rrb)_6$-code, $n\in[20,29]$:}

Refinement of an optimal $(n,11,\llb 2,2,2\rrb)_4$-code, $n\in[20,29]$.

\item {An Optimal $(n,11,\llb 1,1,1,1,1,1\rrb)_7$-code, $n\in[20,26]$:}

Refinement of an optimal $(n,11,\llb 2,1,1,1,1\rrb)_6$-code, $n\in[20,26]$.

\item {An Optimal $(27,11,\llb 1,1,1,1,1,1\rrb)_7$-code:}

\begin{tabular}{l}
010000000002003040506000000 \\
001000000200300004650000000 \\
100000000020030400065000000 \\
020000000300000100000405060 \\
002000000030000010000560004 \\
200000000003000001000056400 \\
000030400000000005001000026 \\
000003040000000500100000602 \\
000300004000000060020000510 \\
000004500000610020000003000 \\
000400050000061002000300000 \\
000040005000106200000030000 \\
345126000000000000000000000 \\
000000123456000000000000000
\end{tabular}

\item {An Optimal $(28,11,\llb 1,1,1,1,1,1\rrb)_7$-code:}

Shorten an optimal $(29,11,\llb 1,1,1,1,1,1\rrb)_7$-code.

\item {An Optimal $(29,11,\llb 1,1,1,1,1,1\rrb)_7$-code:}

Shorten an optimal $(30,11,\llb 1,1,1,1,1,1\rrb)_7$-code.

\item {An Optimal $(30,11,\llb 1,1,1,1,1,1\rrb)_7$-code:}

Shorten an optimal $(31,11,\llb 1,1,1,1,1,1\rrb)_7$-code.

\item {An Optimal $(32,11,\llb 1,1,1,1,1,1\rrb)_7$-code:}

Lengthening of an optimal $(31,11,\llb 1,1,1,1,1,1\rrb)_7$-code.
\end{enumerate}

\section*{Authors' Biographies}

{\bf Yeow Meng Chee} (SM'08) received the B.Math. degree in computer science and
combinatorics and optimization and the M.Math. and Ph.D. degrees in computer science,
from the University of Waterloo, Waterloo, ON, Canada, in 1988, 1989, and 1996, respectively.

Currently, he is an Associate Professor at 
the Division of Mathematical Sciences, School of Physical
and Mathematical Sciences, Nanyang Technological University, Singapore.
Prior to this, he was Program Director of Interactive Digital Media R\&D in the
Media Development Authority of Singapore,
Postdoctoral Fellow at the University of Waterloo and
IBM's Z{\"u}rich Research Laboratory, General Manager of the Singapore Computer Emergency
Response Team, and 
Deputy Director of Strategic Programs at the Infocomm Development Authority, Singapore.
His research interest lies in the interplay between combinatorics and computer science/engineering,
particularly combinatorial design theory, coding theory, extremal set systems,
and electronic design automation.

\vskip 10pt

{\bf Son Hoang Dau} received the BachelorÕs degree in Applied Mathematics and Informatics from the College of Science, Vietnam National University, Hanoi, Vietnam, in 2006
and the M.S. degree in mathematical sciences from the Division of Mathematical Sciences,
Nanyang Technological University, Singapore, where he is currently working towards
the Ph.D. degree.

His research interests are coding theory and combinatorics.

\vskip 10pt 

{\bf Alan C. H. Ling} was born in Hong Kong in 1973. He received the B.Math.,
M.Math., and Ph.D. degrees in combinatorics \& optimization from the University
of Waterloo, Waterloo, ON, Canada, in 1994, 1995, and 1996, respectively.
He worked at the Bank of Montreal, Montreal, QC, Canada, and Michigan
Technological University, Houghton, prior to his present position as Associate
Professor of Computer Science at the University of Vermont, Burlington. His research
interests concern combinatorial designs, codes, and applications in computer
science.

\vskip 10pt 

{\bf San Ling} received the B.A. degree in mathematics from the
University of Cambridge, Cambridge, U.K., in 1985
and the Ph.D. degree in mathematics from
the University of California, Berkeley, in 1990.

Since April 2005, he has been a Professor with the
Division of Mathematical Sciences, School of Physical and
Mathematical Sciences, Nanyang Technological University,
Singapore. Prior to that, he was with the Department of
Mathematics, National University of Singapore. His research fields
include arithmetic of modular curves and application of number
theory to combinatorial designs, coding theory, cryptography and
sequences.

\end{document}